\documentclass[lettersize,journal]{IEEEtran}
\usepackage{amsmath,amsfonts}
\usepackage{algorithmic}
\usepackage{algorithm}
\usepackage{array}
\usepackage[caption=false,font=normalsize,labelfont=sf,textfont=sf]{subfig}
\usepackage{textcomp}
\usepackage{stfloats}
\usepackage{url}
\usepackage{verbatim}
\usepackage{graphicx}
\usepackage{cite}
\usepackage[english]{babel}
\usepackage{times}
\usepackage[T1]{fontenc}
\usepackage{amsfonts}
\usepackage{amssymb}
\usepackage{amsthm}
\usepackage{amsmath}
\usepackage{cases}
\usepackage{bm}
\usepackage{graphicx}
\usepackage{float}
\usepackage{subfig}
\usepackage[section]{placeins}
\usepackage{enumerate}
\newtheorem{theorem}{\bf Theorem}
\newtheorem{remark}{\bf Remark}
\usepackage{color}
\usepackage{lineno}
\modulolinenumbers[5]
\usepackage{url}
\begin{document}

\title{Proof-of-Prospect-Theory: A Novel Game-based Consensus Mechanism for Blockchain}

\author{\IEEEauthorblockN{Yuqi Xie, Changbing Tang, \IEEEmembership{Member, IEEE}, Feilong Lin, \IEEEmembership{Member, IEEE},  Guanrong Chen, \IEEEmembership{Life Fellow, IEEE},  Zhao Zhang, \IEEEmembership{Member, IEEE},  Zhonglong Zheng, \IEEEmembership{Member, IEEE}}

\thanks{This work was partly supported by the National Natural Science Foundation of China (No. 62103375), and the Zhejiang Provincial Natural Science Foundation of China (No. LY22F0300006, 22NDJC009Z).
(Corresponding author: Changbing~Tang, Guanrong~Chen.)\par
Yuqi Xie and Zhao Zhang are with the College of Mathematical Science, Zhejiang Normal University, Jinhua 321004, China (e-mail: hsiehyq@zjnu.edu.cn; zhaozhang@zjnu.cn).\par
Changbing Tang is with the College of Physics and Electronic Information Engineering, Zhejiang Normal University, Jinhua 321004, China (e-mail: tangcb@zjnu.edu.cn).\par
Feilong Lin and Zhonglong Zheng are with the College of Computer Science and Technology, Zhejiang Normal University, Jinhua 321004, China (e-mail: bruce\_lin@zjnu.edu.cn; zhonglong@zjnu.cn).\par
Guanrong Chen is with the Department of Electrical Engineering, City University of Hong Kong, Hong Kong SAR,  China (e-mail: eegchen@cityu.edu.hk).
}
}


\maketitle

\begin{abstract}
Blockchain technology is a breakthrough in changing the ways of business and organization operations, in which the consensus problem is challenging with practical constraints, such as computational power and consensus standard.
In this paper, a novel consensus mechanism named Proof-of-Prospect-Theory (PoPT) is designed from the view of game theory, where the game prospect value is considered as an important election criterion of the block-recorder.
PoPT portrays the popularity of a node in the network as an attribute, which is constituted by the subjective sensibilities of nodes.
Furthermore, the performances of the PoPT and the willingness of ordinary nodes to participate in the consensus are analyzed, exploring fairness, decentralization, credibility, and the motivating ability of the consensus mechanism.
Finally, numerical simulations with optimization of the PoPT consensus mechanism are demonstrated in the scenario of a smart grid system to illustrate the effectiveness of the PoPT.
\end{abstract}

\begin{IEEEkeywords}
Blockchain, Consensus Mechanism, Game Theory, Proof-of-Prospect-Theory, Prospect Theory.
\end{IEEEkeywords}
\section{Introduction}\label{Introduction}
Blockchain technology was first proposed in Satoshi Nakamoto's Bitcoin white paper \cite{nakamoto}, which is a new distributed framework using chained block data structures, consensus mechanisms, and cryptography to store, validate and manipulate data \cite{Mollah}.
Blockchain technology has a remarkable superiority of decentralization and recently generated considerable interest from both academia and industry \cite{blockchainforSmartEnergy, IoT, Jiaheng}.

As a core technique of blockchain, the consensus mechanism focuses on distributed consistency, which has been used in both identification and tampering prevention to maintain decentralized multi-party mutual trust.
So far, many consensus mechanisms have been proposed, which generally can be divided into four categories \cite{IEM}:
Byzantine Fault Tolerance (BFT)-based, Proof-of-Work (PoW)-based, Proof-of-Stake (PoS)-based, and mixed-type consensus mechanisms.
The BFT-based consensus mechanism allows nodes in the network to vote to record the new blocks directly \cite{BFT}, in which the information exchange hinders the scalability of the blockchain.
PoW-based consensus mechanism allocates the recording rights by the hash power competition among the nodes \cite{nakamoto}.
However, massive amounts of wasted electricity made them widely criticized.
In the PoS-based consensus mechanism, nodes compete for the recording rights with the tokens they hold \cite{PoS}, which reduces the resource consumption in the PoW-based consensus mechanism.
Inspired by PoS, the Delegated Proof-of-Stake (DPoS) was proposed \cite{dpos}, where a subset of miners is elected as block-recorder to validate the chain of blocks.
The downside lies in that the PoS-based consensus mechanism tends to be centralized with the centralization of stake.
As for the mixed-type consensus mechanism, it uses unique theories or methods to suit specific scenarios, including Proof-of-Elapsed-Time (PoET) \cite{poet}, Proof-of-Authority (PoA)~\cite{poa}, Proof-of-Trust (PoT) \cite{pot}, etc.
A survey on this topic can be referred to \cite{surveyofBlockchainconsensus, zheng2017overview}.

Usually, in the process of consensus, there is a conflict between individual rationality and group rationality.
Therefore, the key issue is how to fully stimulate individuals with conflicting interests to formulate advantageous strategies in the competition.
Some approaches adopt game theory to handle this problem, which results from that game theory can provide an analytical tool to study the conflict between individual rationality and group rationality in the field of blockchain \cite{gameblockchain, tang}.
In \cite{Incentivizing}, a Stackelberg game was formulated to model the interaction between blockchain user and miner that decides the number of recruited verifiers over wired or wireless networks.
In \cite{qin2018research}, the mining pool selection problem was modeled as a game theory-based risk decision problem, aiming to maximize the likelihood of successful blockchain mining.
A series of game-theoretical models of competition was used in \cite{johnson2014game} to explore the trade-off between investing in additional computing resources and triggering a costly distributed denial-of-service attack.
However, the above studies only regard game theory as a tool for analyzing research problems but rarely design consensus mechanisms in a bottom-up manner from the perspective of risk decision-making.

In a game, players always make decisions based on perceived losses or gains and focus more on the relative utility that they will get.
Generally, they underestimate small probabilities, even if it is possible to lose all of their wealth.
Because of underestimating small probabilities, they end up choosing high-risk options with higher probabilities.
Some studies have considered the mental states of players.
Most representatively, in 1979, Kahneman and Tversky proposed the prospect theory which is one of the most important theories in decision theory \cite{kahneman1979prospect}.
Because of prospect theory, Kahneman won the 2002 Nobel Prize in Economics, which describes how people make decisions when being presented with alternatives that involve risk, probability, and uncertainty.
Prospect theory has been applied on many investigations, such as privacy protection \cite{Privacy}, energy exchange \cite{xiao2014prospect}, and epidemic spreading \cite{9674745}.

Inspired by the prospect theory, this paper designs a game-based consensus mechanism from the view of risk decision-making.
Specifically, game theory is applied as a key component to design a consensus mechanism named Proof-of-Prospect-Theory (PoPT).
In summary, the main contributions of this paper are summarized as follows.
\begin{itemize}
\item We propose a novel blockchain consensus mechanism (named PoPT) based on prospect theory.
    The PoPT consensus mechanism does not need to solve hash puzzles like PoW, which can save computing power costs.

\item We integrate the prospect theory into the PoPT consensus mechanism from the view of decision-making, where several competitive and cooperative nodes with high Prospect Value (PV) are likelier to be elected as block-recorder to pack data into blocks.

\item We study the comprehensive performance of the PoPT consensus mechanism and the willingness of nodes.
    The optimal probabilities for applicant nodes being block-recorder and the block reward are analyzed theoretically to explore the fairness,  decentralization, credibility, and motivating ability of the proposed consensus mechanism.
\end{itemize}

The rest of this paper is organized as follows.
In section \ref{Related}, we summarize the related work of consensus mechanisms.
Section \ref{Preliminaries} introduces some preliminaries about the expected utility theory and the prospect theory.
In section \ref{popt}, the design, optimization, and process of the PoPT consensus mechanism are studied theoretically.
In section \ref{performance}, we analyze the security and energy consumption of the PoPT consensus mechanism.
In section \ref{simulation}, we do several experiments to illustrate the performance of the PoPT consensus mechanism.
Finally, section \ref{conclusion} concludes the paper.
\section{Related Work}\label{Related}
Although blockchain technology has achieved remarkable development, it still faces some problems, such as security threats, low performance, and fork problems \cite{robust}.
The consensus mechanism in blockchain aims to solve the data consistency problem in the presence of faulty nodes in a distributed system.
Until now, consensus mechanisms have been studied for a long time in the field of distributed systems.
In this section, we outline and discuss four kinds of typical consensus mechanisms: BFT-based, PoW-based, PoS-based, and mixed-type consensus mechanisms.

The BFT mechanism is a fault-tolerance mechanism based on the Byzantine problem, which addresses how to reach a consensus with reliable communication when nodes are at risk of failure \cite{BFT}.
In \cite{PBFT}, an improved BFT mechanism, Practical Byzantine Fault Tolerance (PBFT), was proposed to reduce the operational complexity of BFT from exponential to polynomial level, making BFT usable in practice.
Although the traffic of the PBFT mechanism can be reduced by dividing the consensus process into multiple layers, the mechanism still maintains high complexity.
In addition to the traffic, BFT-based consensus mechanisms such as Zyzzyva \cite{zyzzyva}, HoneyBadgerBFT \cite{HoneyBadger}, and BEAT \cite{BEAT} are proposed to strengthen the security and performance of consensus mechanism.
For recent progress on BFT, readers can refer to \cite{Duan}.

The core idea of PoW is that nodes solve the hash puzzle by competition among computational power to gain the right to record blocks.
However, nodes spend power to solve the meaningless hash puzzle resulting in extremely large amounts of excess electricity consumption.
From the perspective of the puzzle, Proof-of-Solution (PoSo) imitated PoW by replacing the meaningless hash puzzle with a meaningful optimization problem, whose solution is hard to find but easy to verify \cite{poso}.
Moreover, PoW also suffers from the fork, low efficiency, and other problems, and much work has been done to solve these problems.
GHOST \cite{sompolinsky2015secure} and Ethash \cite{buterin2014next} consensus introduced the uncle block to solve the fork and shorten the consensus time interval.
Bitcoin-NG divided the blocks into PoW blocks and normal blocks to improve efficiency \cite{eyal2016bitcoin}.

In the PoS consensus mechanism, block-recorders are chosen based on the number of tokens they hold, in which the stake amount replaces the work nodes do in the PoW.
Based on the PoS, the EOS blockchain platform proposed the DPoS consensus mechanism, where a certain number of representatives are elected through voting to engage \cite{dpos}.
From the perspective of indicator, \cite{BentovLMR14} designed the consensus mechanism, where the stake is replaced with vitality.
In general, PoS-based consensus mechanisms tend to be centralized with the centralization of indicators, which is a key issue that should be carefully considered.

Mixed-type consensus mechanisms are designed from specific scenarios \cite{poet, poa, pot}.
For example, PoET applied the random waiting mechanism to simulate the time consumption required to solve Hash in PoW, and the node with the shortest waiting time gets the priority to publish blocks (similar to the earliest solution of Hash problem in PoW) \cite{poet}.
Proof-of-State-Velocity (PoSV) was designed specifically for the digital social currency Reddcoin, which should only be evaluated in the Reddcoin ecosystem \cite{posv}.
Furthermore, to enhance the scalability of blockchain, a new data structure has sprung up based on directed acyclic graphs (DAGs) in recent years \cite{benvcic2018distributed, churyumov2016byteball}.
However, the validation time of the DAG-based consensus mechanism becomes unstable with the DAG graph expansion. Moreover, the asynchronous mechanism of DAG does not guarantee a global order, which affects the credibility of the service.

Although there has been a lot of research on consensus mechanisms, the design of consensus mechanisms is still a hot topic.
Different blockchain systems are designed with different purposes in mind, which leads to different consensus mechanisms with different focuses.
This paper designs the PoPT consensus mechanism from the perspective of game theory, in which the probability of nodes becoming the block-recorder is derived based on prospect theory.
In fact, prospect theory reflects users' decisions under various psychological expectations and personalities.
According to prospect theory, the PoPT consensus mechanism builds a behavioral model of competition and cooperation among system users, which was seldom studied before.
\section{Preliminaries}\label{Preliminaries}
\subsection{Expected utility theory in game theory}
Game theory is used to analyze the interactions among users.
Consider a scenario with $I$ users, where user $i$ ($i \in I$) is faced with multiple options.
Assume there are only options $X$ and $Y$.
When selecting $X$, there are $M$ possible events with values $x_m$ ($m \in M$), occurring with probability $\rho_{x_m}$.
For another option $Y$, there are $N$ possible events with values $y_n$ ($n \in N$), occurring with probability $\rho_{y_n}$.
Mathematically, for a discrete variable $X$ the expected value $E[X]$ is given by $\sum_{m=1}^{M} x_m\rho_{x_m}$.
Similarly, for $Y$, $E[Y]$ is given by $\sum_{n=1}^{N} y_n\rho_{y_n}$.
User $i$ makes a choice between $X$ and $Y$ based on $E[X]$ and $E[Y]$.
The interactions among users are described as follows: (i) Player: User $i$ in blockchain network, $i \in I$; (ii) Strategy: Making choices among all options, $X$ or $Y$; (iii) Payoff: The utility of players under different strategic combinations, $E[X]$ or $E[Y]$.

Expected utility theory is used to analyze the decision-making of individuals but with limitation in accurately determining the subjective utility function.
Kahneman and Tversky in \cite{kahneman1979prospect} fully explained that choices among prospects exhibit several pervasive effects which are inconsistent with the tenets of expected utility theory.
Fortunately, the relationship between risk and return can be empirically studied by prospect theory.

\subsection{Prospect theory in game theory}\label{Prospecttheoryingametheory}
Prospect theory is defined by the value function and weighting function.
The value function substitutes value in the expected utility theory.
When players make choices, they make comparisons with other specific references intentionally or unintentionally.
When the reference is different, even for the same outcome, players perceive a difference in benefits.
The value function is S-shaped, which is convex upward for ``gain'' and concave downward for ``loss''.
As it moves to both ends, the change of direction decreases in sensitivity.
Compared with the case where the ``gain'' is small on the right of the reference, the smaller the ``loss'' on the left of the reference point, the steeper the function.

Taking option $X$ with values $x_m$ $(m \in M)$ for example, the value function is set as \cite{tversky1992advances}:
\begin{equation}
v(x_m)= \begin{cases}\left(x_m-x_{0}\right)^{\alpha}, & \text { if } x_m \geq x_{0} \\ -\lambda\left(x_{0}-x_m\right)^{\beta}, & \text { if }  x_m \leq x_{0},\end{cases}
\label{value}
\end{equation}
where $x_{0}$ is the reference for the decision maker, and $\alpha$, $\beta$ are coefficients reflecting the attitudes toward risk.
$0<\alpha,\beta<1$, for which the greater they are, the more adventurous the decision maker is.
$\lambda$ is the loss aversion coefficient, where if $\lambda>1$ the decision maker is more sensitive to loss.
In addition, the weighting function $\pi(\rho)$ replaces the probability $\rho$ in expected utility theory, which represents a certain subjective judgment made by the decision maker based on probability.
Here, $\pi(\rho)$ is a monotonic function of $\rho$.
For small probabilities, it always gives a very large weight, that is, $\pi(\rho)>\rho$; for large probabilities, the weight is always very small, that is, $\pi(\rho)<\rho$.
Taking option $X$ with probabilities $\rho_{x_m}$ for example, the weighting function is set as~\cite{prelec1998probability}:
\begin{equation}
	\pi(\rho_{x_m})=\exp(-(-\ln \rho_{x_m})^{\phi}).
	\label{weight}
\end{equation}
Then, the PV of option $X$ can be obtained by
$
	\label{PV}
	P[X]=\sum_{m=1}^{M} v(x_m)\pi(\rho_{x_m})
$, and the PV of option $Y$ can be obtained by
$
P[Y]=\sum_{n=1}^{N} v(y_n)\pi(\rho_{y_n})
$.
User $i$ makes a choice between $X$ and $Y$ based on $P[X]$ and $P[Y]$.
\section{PoPT: Proof-of-Prospect-Theory}\label{popt}
The consensus mechanism of blockchain mainly deals with the election of block-recorder and the consistency of distributed ledgers in a decentralized environment.
This section integrates the prospect theory into the design of the consensus mechanism, for which the design concept uses the indicator of PV as the election criteria for the consensus block-recorder.
On this basis, the optimization and the processes of the PoPT consensus mechanism are developed.
For readability, the main symbols are summarized in Table \ref{tab:Notations}.

\begin{table}[!t]
\caption{SYMBOL \quad DEFINITION\label{tab:Notations}}
\begin{tabular}{c|l}
\hline
\textbf{Symbol} & \textbf{Definition}\\
\hline\
$t$ & Block slot indicator\\
$v(.)$ & Value function of prospect theory \\
$\alpha, \beta, \lambda$ & Parameters of the value function\\
$\pi(.)$ & Weighting function of prospect theory\\
$\phi$ & Parameter of weighting function\\
$pv_i$ & PV of applicant node $i$\\
$l$ & Loss factor \\
$\mathcal{F}, \mathcal{D}, \mathcal{C}$ & Fairness, decentralization, and credibility of PoPT\\
$\mu_1, \mu_2$ & Weights of PoPT's fairness and decentralization\\
$R$ & Block reward \\
$w_i$ & Willingness of ordinary node $i$ to be an applicant node\\
$u_i$ & Utility of ordinary node $i$ in the consensus process\\
$u_{i}^{0}$ & Expected utility of ordinary node $i$ in the consensus process\\
$p_i$ & Probability of ordinary node $i$ becoming block-recorder\\
$k$ & Commission charge rate\\
$\theta_{i}$ & Average transaction volume of  ordinary node $i$\\
$\mathbb{A}$ & The set of applicant nodes\\
$\mathbb{O}$ & The set of ordinary nodes\\
\hline
\end{tabular}
\end{table}

\subsection{Design of PoPT}\label{design}
There are three types of nodes in the proposed PoPT: ordinary nodes, applicant nodes, and block-recorder.
Initially, all nodes are ordinary nodes.
Nodes that want to be block-recorder submit an application with their PV to the blockchain system.
The blockchain system calculates the probability of each applicant node becoming a block-recorder according to the purpose of the system, and the block-recorder is selected from the applicant nodes by this probability.

Nodes in the network interact according to the purpose of the system.
For example, nodes in an energy trading system always exchange information about the energy price and the amount of energy to be purchased.
After exchanging information, they decide whether to agree to participate in the transaction, and the PV quantifies the nodes' decisions by prospect theory.

Dividing time into block slots, consider the scenario where node $j$ interacts with node $i$ at the $t$-\emph{th} block slot.
Taking trading as an example of interaction, assume that trading price of this trade is $u_{ij}$, while the expected price of seller $j$ for this trade is $u_{ij}^0$.
Denote the willingness of buyer $i$ for this trade by $\rho_{ij}(t)$.
Substitute $u_{ij}(t)$ and $u_{ij}^0$ into Eq. \eqref{value} for $x_m$ and $x_0$, respectively.
Substitute $\rho_{ij}(t)$ for $\rho_{x_m}$ in Eq. \eqref{weight}.
Referring to the prospect theory-based payoff mentioned in subsection \ref{Prospecttheoryingametheory}, we denote the following equation as $p_{ij}(t)$ which means the PV of seller $i$ to buyer $j$ at the $t$-\emph{th} block slot:
\begin{equation}\label{pij}
 p_{ij}(t)=v\left(u_{ij}(t)\right) \pi\left(\rho_{ij}(t)\right).
\end{equation}
According to the implication of prospect theory, $p_{ij}(t)$ describes the subjective utility in the mind of buyer $j$ when trading with seller $i$ .
A higher $p_{ij}(t)$ indicates that buyer $j$ is more satisfied with seller $i$ and more willing to trade with seller $i$.
If all buyers are willing to interact with seller $i$, then it is reasonable to let seller $i$ participate in data recording.
For $I$ sellers and $J$ buyers, there is a matrix
\begin{equation}
	PV^{\prime}(t)=(pv^{\prime}_{ij}(t))_{I \times J}=\left[\begin{array}{ccc}
		pv^{\prime}_{11}(t) & \cdots & pv^{\prime}_{1 J}(t) \\
		\vdots & \ddots & \vdots \\
		pv^{\prime}_{I 1}(t) & \cdots & pv^{\prime}_{I J}(t)
	\end{array}\right].
\label{4}
\end{equation}
Here, in order to avoid unfairness caused by buyers' expectations, elements in the matrix \eqref{4} are normalized with the 2-norm,
i.e., $pv_{i j}^{\prime}(t)=pv_{i j}(t) /\left\|pv_{. j}(t)\right\|_2$.
Considering that the character of the node will change over time, take only the values within the last $T$ periods and a loss factor $l$ into account, $pv_{i}(t)$ is the accumulated PV of seller $i$ until the $t$-\emph{th} block slot:
\begin{equation}\label{pit}
	pv_{i}(t)=\sum_{k=t-T+1}^{t} l^{t-k} \frac{1}{J} \sum_{j=1}^{J} pv_{i j}^{\prime}(k).
\end{equation}
The PV describes the status of nodes in the blockchain system, and it is publicly recorded in the blockchain.
If node $i$ is an applicant node, then $pv_{i}(t)$ generally determines whether applicant node $i$ can be the block-recorder.
If no confusion occurs, we omit $t$.
\subsection{Optimization of PoPT}\label{Optimization}
This subsection optimizes the consensus mechanism by determining the optimal chosen probability and block reward.
Optimization problem \textbf{\emph{P\bm{$1$}}} is to set the probability of being the block-recorder for each application node, while ensuring the fairness, decentralization, and credibility of PoPT.
Optimization problem \textbf{\emph{P\bm{$2$}}} refers an optimal reward to motivate ordinary nodes to participate in the consensus process based on the probability in \textbf{\emph{P\bm{$1$}}}.
We specifically design the optimization smart contract to handle these two problems with the Grey Wolf Optimization (GWO) algorithm.
\subsubsection{Optimal probability}
The probability of an applicant node becoming the block-recorder is determined based on its PV, which ensures that the PoPT consensus mechanism is fair, decentralized, and reliable.

From the fairness point of view, the higher the prospect value of an applicant node, the higher the probability of becoming the block-recorder should be.
Denoting the set of applicant nodes as $\mathbb{A}$, the share of the PV of applicant node $i$ in the overall PV is $a_i$, that is, \begin{equation}\label{ai}
a_i=pv_i/\sum_{j \in \mathbb{A}} pv_j.
\end{equation}
And the probability of applicant node $i$ becoming the block-recorder is $p_i$.
We refer to the expectational fairness to describe the fairness $\mathcal{F}$ of PoPT \cite{richer}:
\begin{equation}
\mathcal{F}=1-\frac{\sum\limits_{i \in \mathbb{A}}\left|a_i-p_i\right|}{N_a \cdot \max\limits_{i \in \mathbb{A}}\left|a_i-p_i\right|}.
\end{equation}
$\mathcal{F} \in [0,1]$ and the higher the $\mathcal{F}$, the fairer the PoPT.
The greater the $\mathcal{F}$, the applicant nodes' probabilities of becoming the block-recorder are closer to their PV shares.
If we only consider fairness, some nodes will deliberately manipulate their PVs to widen the gap between their probability of being the block-recorder and the probability of other nodes, which will make the consensus process centralized.
Thereby, we describe the decentralization $\mathcal{D}$ of PoPT by the Nakamoto index which describes the minimum number of nodes required to make the sum of probabilities greater than 50\% \cite{lin2021measuring}:
\begin{equation}
\mathcal{D}=\frac{k^*}{|\mathbb{A}|},
\end{equation}
where
\begin{equation}
k^*=\min \left\{k \in\left[1,2, \ldots, |\mathbb{A}|\right]: \sum_{i=1}^k p_i \geq 0.5\right\}.
\end{equation}
The greater the $\mathcal{D}$, the PoPT consensus mechanism is more decentralized.
In addition to fairness and decentralization, we define the credibility of PoPT as
\begin{equation}
\mathcal{C}=\frac{N_a \cdot \sum\limits_{i \in \mathbb{A}} p v_i p_i}{\sum\limits_{i \in \mathbb{A}} p v_i},
\end{equation}
which measures the expectations of PV for the selected block-recorder.
The greater the $\mathcal{C}$, the larger the PV of the block-recorder.
Recalling the definition of PV, we observe that the larger $\mathcal{C}$, the more reliable the consensus mechanism.

The optimal probability $\mathbf{p}=\left\{p_1, \ldots p_i, \ldots, p_{|\mathbb{A}|}\right\}$ is set to balance fairness, decentralization, and credibility.
It is determined by the following optimization problem.

\textbf{\emph{P\bm{$1$} (Optimal probability)}}
\begin{equation}
\begin{aligned}
& \mathbf{p}^*=\underset{\mathbf{p}}{\arg \max } \quad \frac{\mu_1^2+\mu_2^2+\left(1-\mu_1-\mu_2\right)^2}{\left(\frac{\mu_1^2}{\mathcal{F}}+\frac{\mu_2^2}{\mathcal{D}}+\frac{\left(1-\mu_1-\mu_2\right)^2}{\mathcal{C}}\right)} \\
& \text { s.t. } 0 \leq p_i, \forall i \in \mathbb{A} \\
& \\
& \qquad \sum_{i \in \mathbb{A}} p_i=1,
\end{aligned}
\end{equation}
where $\mu_1$ and $\mu_1$ are the weights of the PoPT's fairness and decentralization.
The weights satisfy $0 \leq \mu_1, \mu_2 \leq 1$ and $\mu_1+\mu_2 \leq 1$.
To describe PoPT universally, we denote the function optimized by \textbf{\emph{P\bm{$1$}}} as $\mathcal{O}$ which is the comprehensive performance of PoPT:
\begin{equation}\label{O}
\mathcal{O}=\frac{\mu_1^2+\mu_2^2+\left(1-\mu_1-\mu_2\right)^2}{\left(\frac{\mu_1^2}{\mathcal{F}}+\frac{\mu_2^2}{\mathcal{D}}+\frac{\left(1-\mu_1-\mu_2\right)^2}{\mathcal{C}}\right)}.
\end{equation}

To solve \textbf{\emph{P\bm{$1$}}}, the GWO algorithm is applied, which is one of the population-based meta-heuristics algorithm in nature.
First, randomly initialize and hierarchize the population of grey wolves.
Then the coefficient vectors of upper-class wolves are updated to encircle the prey.
Next, the lower-class wolves update their positions in accordance with the upper-class wolves.
The wolves keep exploring until they are close enough to the prey and finally the wolves attack the prey, which means the optimal solution is found.
\subsubsection{Optimal block reward}
Too few nodes participating in the consensus mechanism will put the system at risk.
Higher block reward facilitates nodes to participate in the consensus mechanism.
However, block reward $R$ comes from the commission charge submitted by all nodes.
For these reasons, $R$ should not be too low or too high.

Denoting the willingness of ordinary node $i$ to apply to be an applicant node as $w_{i}$ which is related to the difference between the node's expected utility and the actual utility.
Inspired by the value function of prospect theory in Eq. \eqref{value}, we set
\begin{equation}
w_{i}= \begin{cases}\left(u_{i}-u_{i}^{0}\right)^{\alpha}, & \text { if } u_{i} \geq u_{i}^{0} \\ -\lambda\left(u_{i}^{0}-u_{i}\right)^{\beta}, & \text { if }  u_{i} \leq u_{i}^{0},\end{cases}
\end{equation}
where $u_{i}^{0}$ is the expected utility of ordinary node $i$ in the consensus process, and $u_{i}$ is the actual utility that equals the reward $\pi(p_{i}) R$ from the consensus process minus the commission charge $c_{i}$, i.e.
\begin{equation}
u_{i}=\pi(p_{i}) R-c_{i}.
\end{equation}
Here, $\pi(p_{i})$ is ordinary node $i$'s intuition about the objective probability $p_{i}$, that is, ordinary node $i$ thinks it has a probability $\pi(p_{i})$ of getting the block reward, but the probability that it actually gets the block reward is $p_{i}$.
Formally, we set
\begin{equation}\label{pip}
	\pi(p_{i})=\exp({-(-\ln p_{i})^{\phi_i}}),
\end{equation}
which is set based on the weighting function of prospect theory in Eq. \eqref{weight} and provides a more realistic representation of intuition.
$\phi_i$ is the rationality of the ordinary node $i$.
The larger $\phi_i$ is, the lower the rationality of ordinary node $i$.
When $p_{i}$ is relatively large, a larger $\phi_i$ indicates that node $i$ will exaggerate its probability of getting the block reward.
Similarly, when $p_{i}$ is relatively small, a larger $\phi_i$ indicates that it will shrink its probability of getting the block reward.

The commission charge of node $i$ is set as
$c_{i}=kR \theta_{i}$,
where $k$ is the commission charge rate.
The higher the $R$ is, the more commission charge the nodes need to pay.
$\theta_{i}$ is the average volume of transactions that node $i$ is involved.
Denoting the set of ordinary nodes as $\mathbb{O}$, the optimal $R$ is set to maximize $\sum\limits_{i \in \mathbb{A}}w_{i}$.
Then we can get an optimization problem.

\textbf{\emph{P\bm{$2$} (Optimal Block Reward)}}
\begin{equation}
\begin{aligned}
&R^{*}=\arg \max _{R} \sum_{i \in \mathbb{O}}w_{i}
\end{aligned}
\end{equation}
\smallbreak
\begin{theorem}\label{th1}
For the ordinary node with positive utility in PoPT, there is a unique optimal block reward $R^{*}$ to \textbf{P\bm{$2$}}, which satisfies
\begin{equation}\label{R}
R^{*}=\frac{u_{i}^{0}}{\pi(p_{i})-k \theta_{i}}.
\end{equation}
\end{theorem}

\begin{proof}
Node $i$ gains a positive utility means
\begin{equation}
    u_{i}=\pi(p_{i}) R-c_{i}=\pi(p_{i})R-kR \theta_{i}>0.
\end{equation}
Because $R\geq 0$, one can obtain
$$k<\frac{\pi(p_{i})}{\theta_i}.$$
The set of ordinary nodes satisfying $u_{i} \geq u_{i}^{0}$ is denoted as $\mathbb{O}_1$,  other ordinary nodes are in set $\mathbb{O}_2$, then
\begin{equation}\label{part}
\sum_{i \in \mathbb{O}}w_{i}=\sum_{i \in \mathbb{O}_1}w_{i}+\sum_{i \in \mathbb{O}_2}w_{i}.
\end{equation}
For the first term on the right side of Eq. \eqref{part}, one can find the second-order partial derivative of $\sum\limits_{i \in \mathbb{O}_1}w_{i}$ to $R$ is
\begin{equation}
\frac{\partial^{2} \sum\limits_{i \in \mathbb{O}_1}w_{i}}{\partial R^{2}}=\sum\limits_{i \in \mathbb{O}_1}\alpha(\alpha-1)\left(u_{i}-u_{i}^{0}\right)^{\alpha-2} \cdot\left(\pi(p_{i})-k \theta_{i}\right)^{2}.
\end{equation}
Because $0<\alpha<1$ and $u_{i}>u_{i}^{0}$, one has $\frac{\partial^{2} \sum\limits_{i \in \mathbb{O}_1}w_{i}}{\partial R^{2}} < 0$.
Because the feasible region of $R$ is a convex set, there is a unique optimal block reward $R_{i}$ to \textbf{\emph{P\bm{$2$}}}.
Set the first-order partial derivative of $w_{i}$ to $0$, that is
\begin{equation}
\frac{\partial \sum\limits_{i \in \mathbb{O}_1}w_{i}}{\partial R}=\sum\limits_{i \in \mathbb{O}_1}\alpha\left(u_{i}-u_{i}^{0}\right)^{\alpha-1} \cdot\left(\pi(p_{i})-k \theta_{i}\right)=0.
\end{equation}
Because $k<Q_i/\theta_i$, $\pi(p_{i})>0$, and $\alpha >0$,
\begin{equation}
\alpha\left(u_{i}-u_{i}^{0}\right)^{\alpha-1} \cdot\left(\pi(p_{i})-k \theta_{i}\right) \geq 0, \forall i \in \mathbb{O}_1.
\end{equation}
Thus,
\begin{equation}
\pi(p_{i})R-kR \theta_{i}-u_{i}^{0} = 0, \forall i \in \mathbb{O}_1.
\end{equation}
Eq. \eqref{R} is satisfied for $i \in \mathbb{O}_1$.
For the second term on the right side of Eq. \eqref{part}, because $u_{i} \leq u_{i}^{0}$ and $w_{i}$ increases with increasing $R$,
 one can obtain
$$R^{*}\leq\frac{u_{i}^{0}}{\pi(p_{i})-k \theta_{i}}, \forall i \in \mathbb{O}_2.$$
Thus, Eq. \eqref{R} is satisfied for $i \in \mathbb{O}_2$ as well.
\end{proof}

\begin{remark}
For different ordinary node $i$, $\frac{u_{i}^{0}}{\pi(p_{i})-k \theta_{i}}$ are different.
And in order to maximize the willingness of more nodes, the final value of $R^{*}$ is taken as the median of all $\frac{u_{i}^{0}}{\pi(p_{i})-k \theta_{i}}$.
\end{remark}
\begin{remark}
For different ordinary node $i$, $\frac{\pi(p_{i})}{\theta_i}$ are different.
Based on the principle of individual rationality, nodes should get positive utility for participating in the consensus process, which results in $
k < \min\limits_{i \in \mathbb{O}} \frac{\pi(p_{i})}{\theta_{i}}
$.
\end{remark}

\subsection{Process of PoPT}
In the process of PoPT, we design two smart contracts to execute different missions.
Information smart contract is dedicated to publishing commission charge rate $k$ and the block reward $R$ for ordinary nodes.
Optimization smart contract is used to solve two optimization problems in subsection \ref{Optimization}.

The process of PoPT is illustrated in Fig. \ref{Process} in detail.
\begin{figure*}[!t]
\centerline{\includegraphics[scale=0.42]{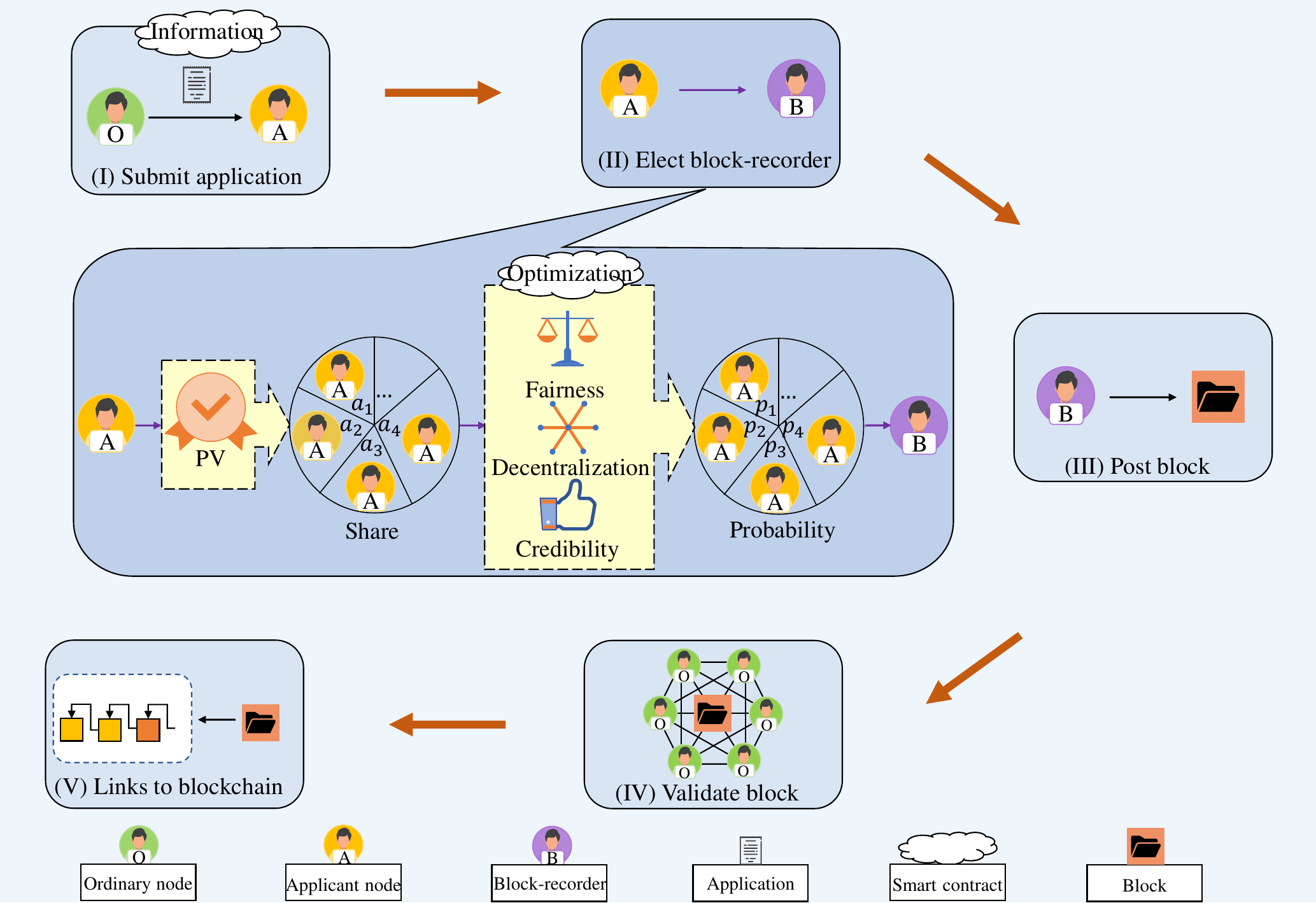}}%
\caption{\small{The process of PoPT.}}
\label{Process}
\end{figure*}
The specific steps are as follows:
\begin{enumerate}[I)]
\item \textbf{Submit application.}
Ordinary nodes receive information from the information smart contract and submit applications to become applicant nodes if they are willing to become block-recorders;
\item \textbf{Elect block-recorder.}
The PV and the PV share of applicant nodes are calculated according to Eq. \eqref{pit} and Eq. \eqref{ai}.
Then, the optimization smart contract calculates the probabilities of application nodes becoming the block-recorder, which considers the fairness, decentralization, and credibility of PoPT.
The block-recorder is elected from applicant nodes according to their probabilities of becoming the block-recorder;
\item \textbf{Post block.}
The block-recorder packs all the transaction information collected during this time slot into a new block and publishes the block in the network;
\item \textbf{Validate block.}
The nodes that receive the block verify the block and send it to all their neighbors if the information is accurate; if not, the block is orphaned;
\item \textbf{Link to blockchain.}
When the block passes the validation of all nodes in the blockchain network, the block is linked to the blockchain, and the block-recorder gains the block reward.
\end{enumerate}

\section{Performance Analysis}\label{performance}
Security and low energy consumption are two fundamental properties of consensus mechanisms.
In this section, we analyze the PoPT's security and energy consumption from its design and operation.
\subsection{Security}
We consider credibility and decentralization when setting the probability to qualify applicant nodes for block-recorder.
In fact, high credibility and decentralization mean fewer malicious nodes.
On one hand, credibility reflects the PV of the chosen block-recorder, and high credibility means that the chosen block-recorder has a high PV.
As we mentioned when formulating the PV, applicant nodes with high PV are adored by the nodes in the network and can be trusted to block data.
Therefore, high credibility guarantees the security of PoPT consensus mechanism to some extent.
On the other hand, applicant nodes may conspire to form a clique to win the block reward, more egregious, manipulate the consensus process.
Fortunately, decentralization describes the minimum number of nodes required to make the sum of probabilities greater than 50\%, and high decentralization ensures that more applicant nodes are needed to form such a clique.
Thus, high decentralization secures the PoPT consensus mechanism from the perspective of preventing collusion.
\begin{figure*}[htbp]
\centering
\subfloat[]{\includegraphics[scale=0.35]{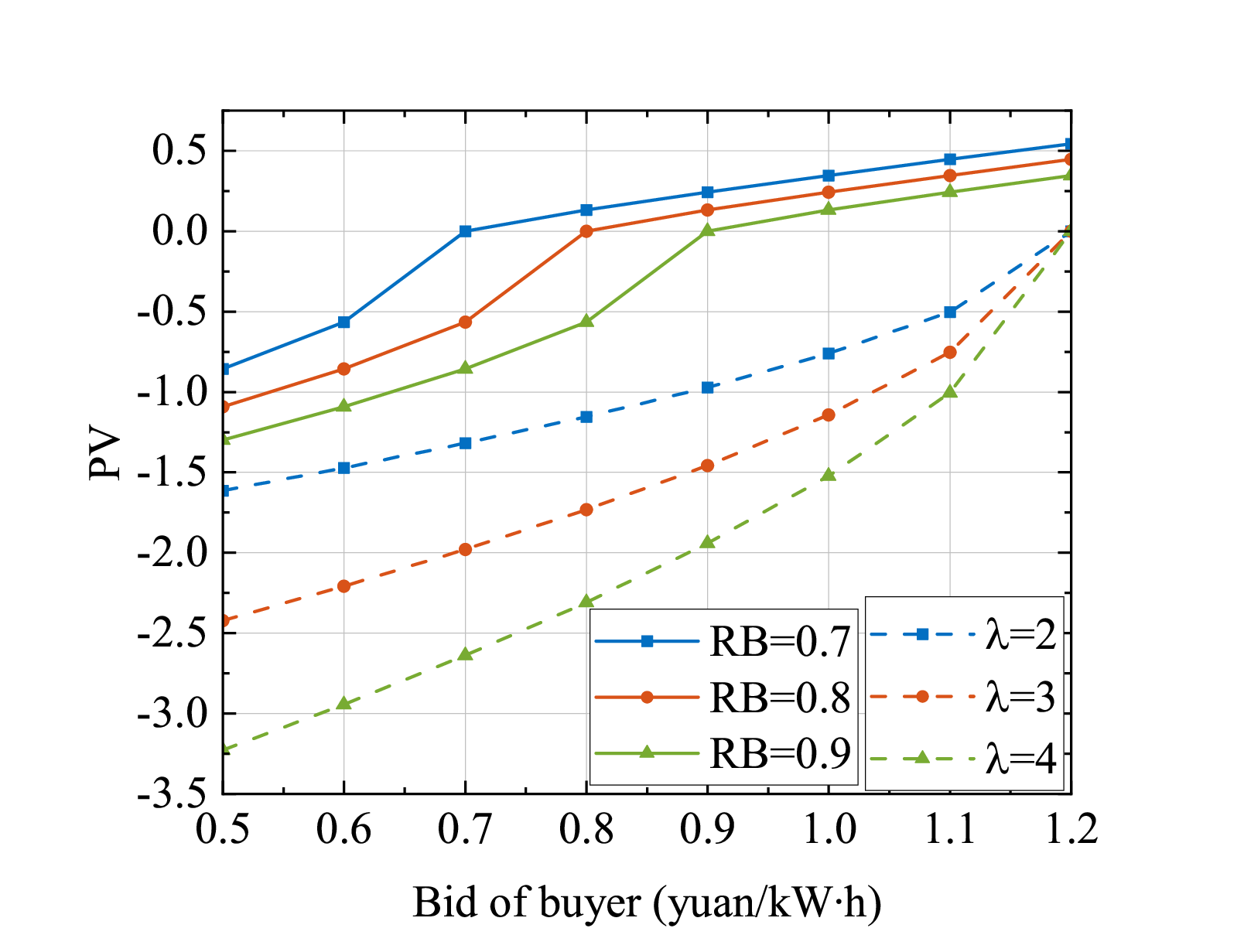}%
\label{fig_2_case1}}
\subfloat[]{\includegraphics[scale=0.35]{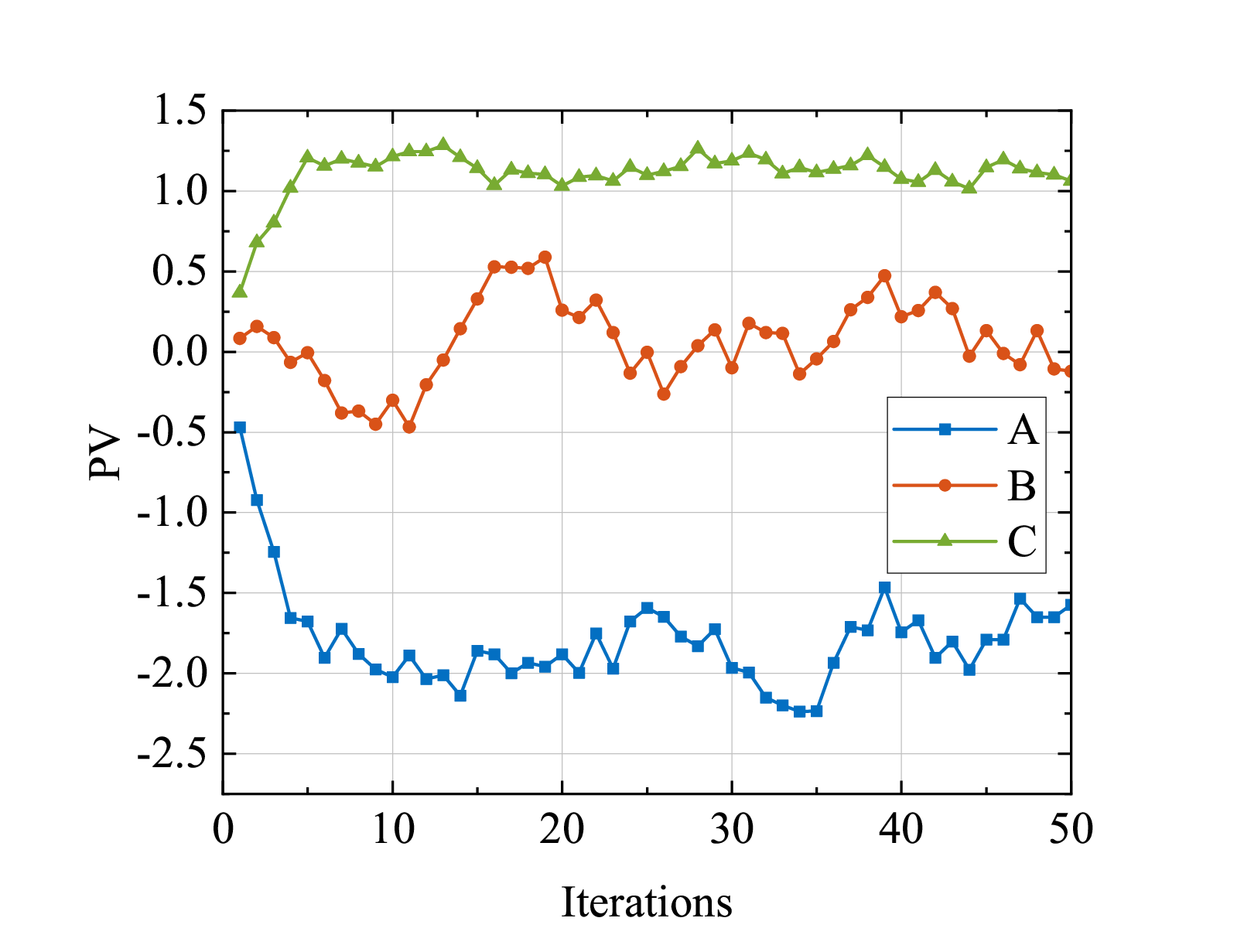}%
\label{fig_2_case2}}
\caption{\small{The variation of PV. (a)The effect of seller's personality on buyers' PV at different bids. (b)The comparison of PV among three types of buyers.}}
\end{figure*}
\begin{figure}[tbp]
\centerline{\includegraphics[scale=0.35]{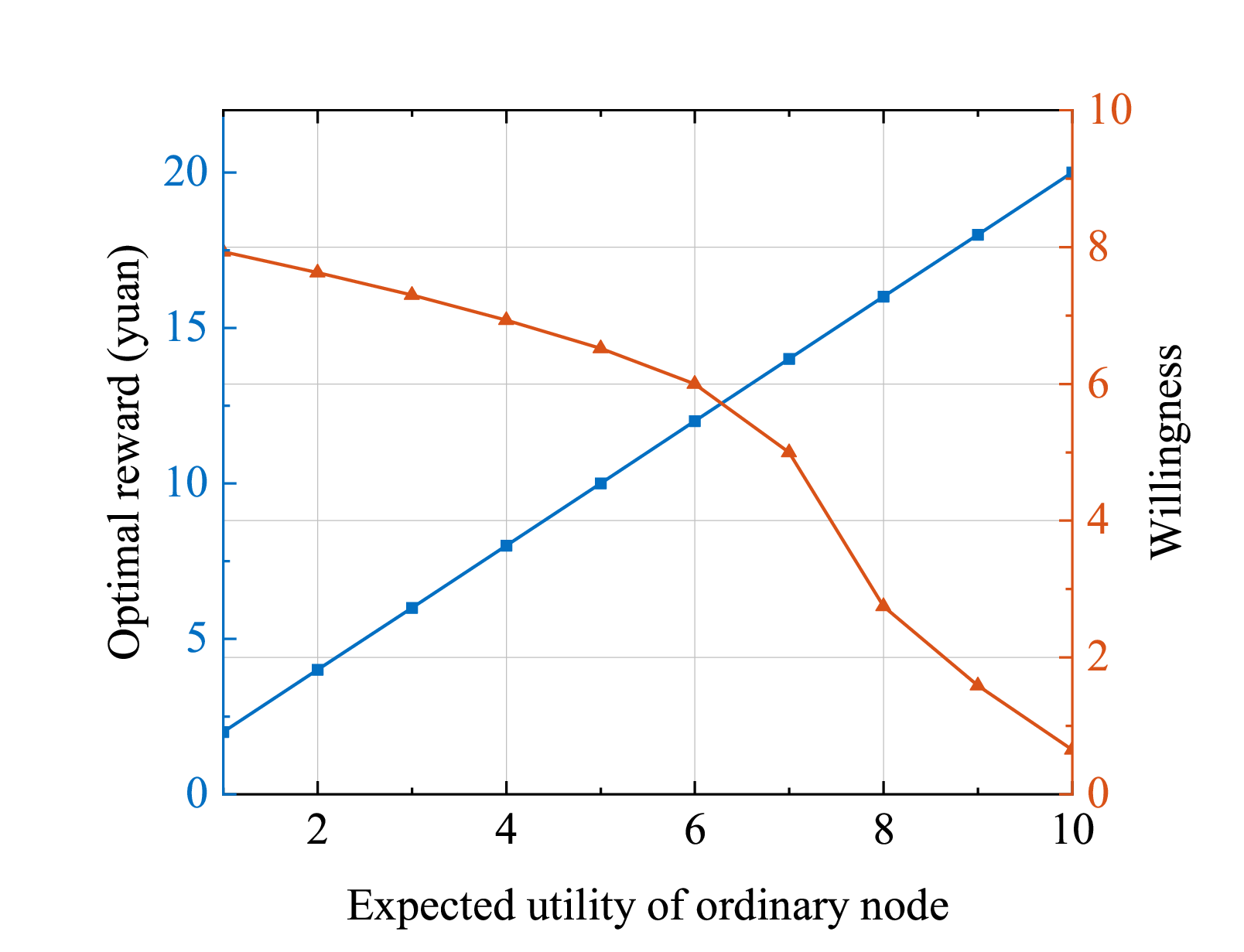}}%
\caption{\small{The optimal reward and willingness of ordinary nodes to be application nodes on expected utility of ordinary nodes.}}
\label{fig_3}
\end{figure}
\begin{figure}[tbp]
\centerline{\includegraphics[scale=0.35]{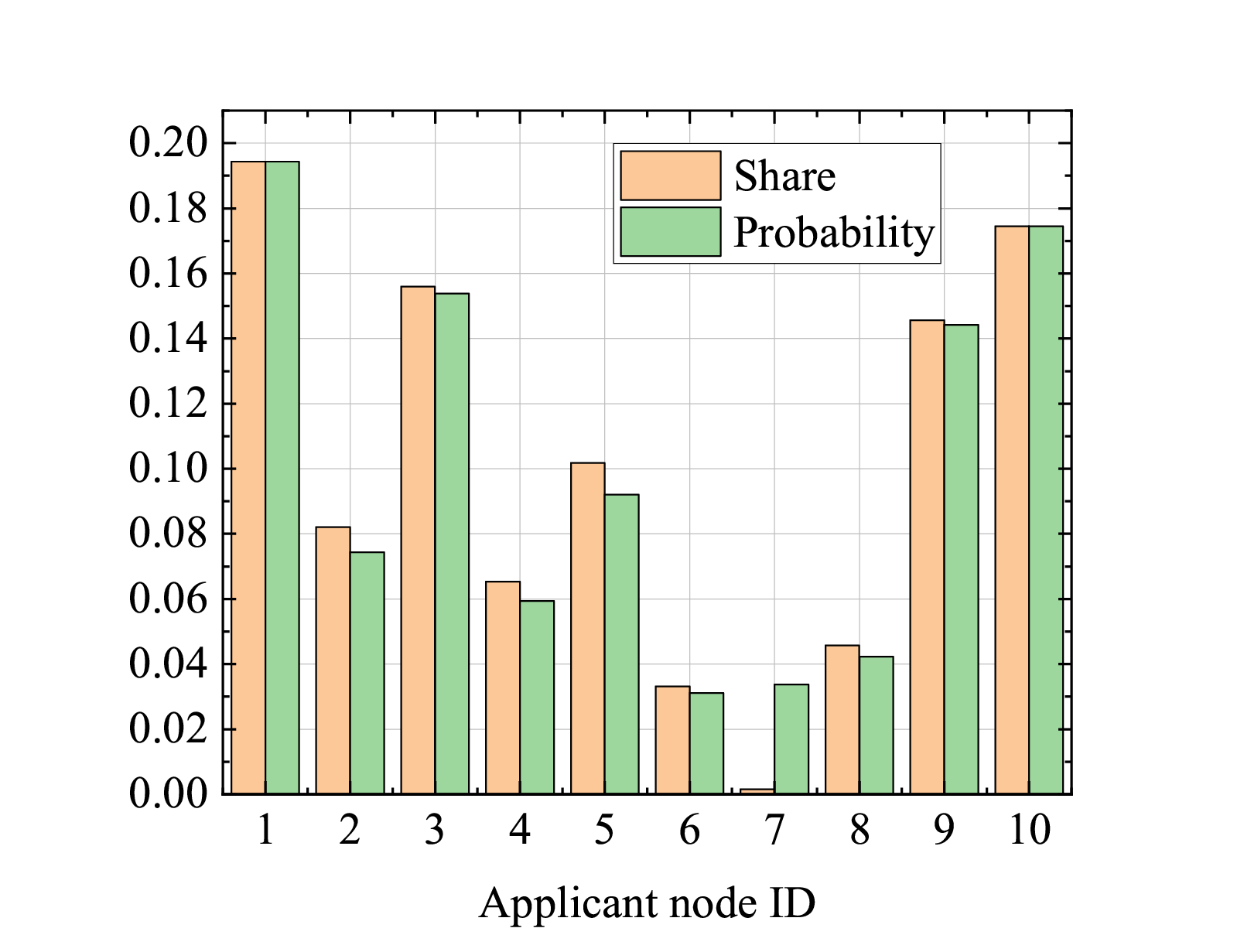}}%
\caption{\small{The comparison of the applicant nodes' PV share and the probability of becoming the block-recorder.}}
\label{fig_4}
\end{figure}
\subsection{Energy consumption}
The energy consumption is reduced by avoiding solving the meaningless puzzle as PoW does.
The PoPT consensus mechanism elects the block-recorder by the probability gained in the optimization problem \textbf{\emph{P\bm{$1$}}}, which only takes very little energy.
The PoPT consensus mechanism first determines the probability of each applicant node by solving the optimization problem \textbf{\emph{P\bm{$1$}}}.
Then the PoPT determines the block-recorder by this probability.
The optimization problem \textbf{\emph{P\bm{$1$}}} has only one objective function and two linear constraints, which is not difficult to solve.
Some solvers in MATLAB software can be used to solve this problem directly, such as \emph{fmincon}.
In addition, some heuristic algorithms can also be used to solve this problem, such as the GWO algorithm and particle swarm optimization.
We find that the computational power and time required to solve this problem are low, which indicates that our proposed PoPT consensus mechanism is low energy consumption.
\section{Simulation Analysis}\label{simulation}
\subsection{Simulation environment}
\begin{figure*}[htbp]
\centering
\subfloat[]{\includegraphics[scale=0.34]{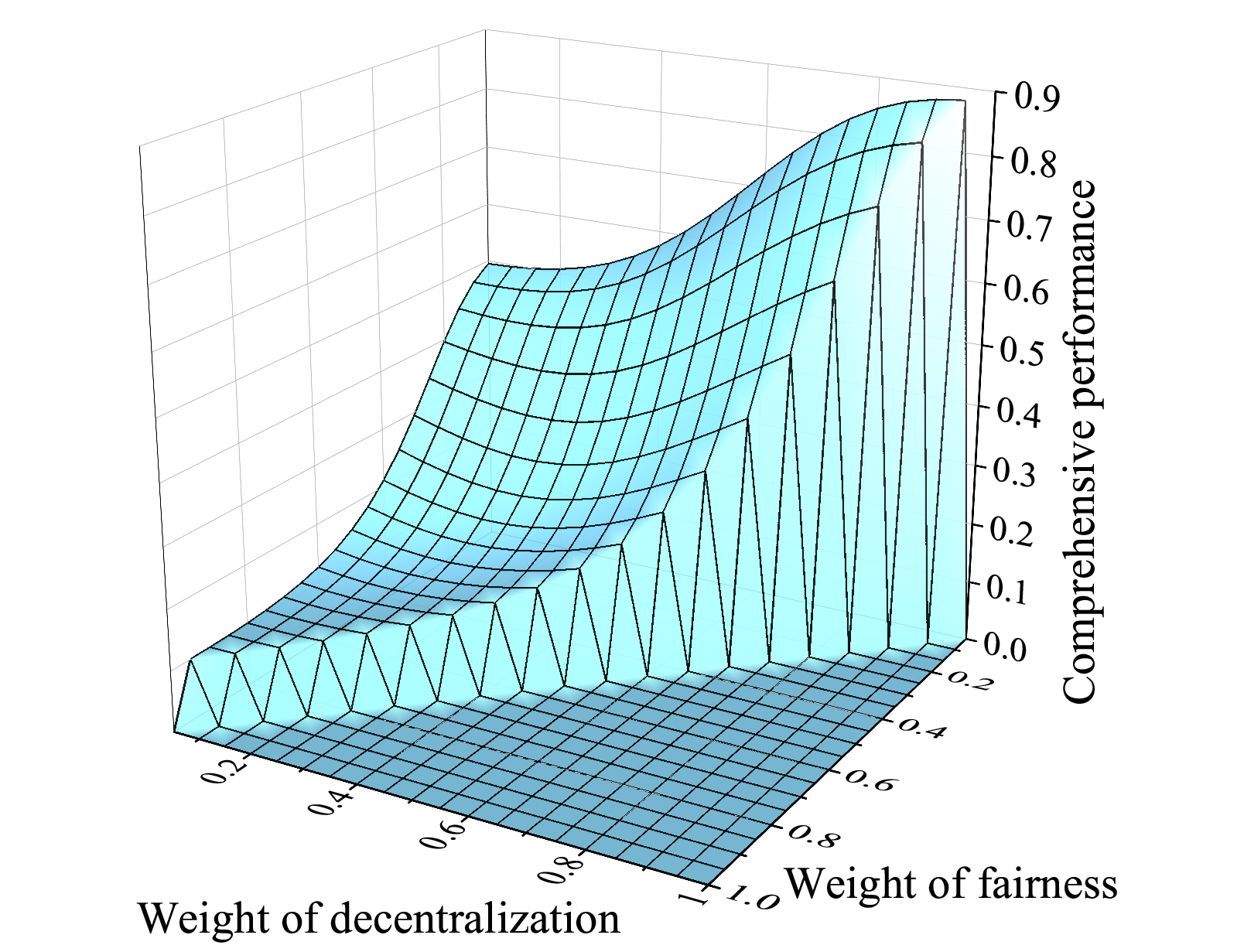}%
\label{fig_5_case1}}
\subfloat[]{\includegraphics[scale=0.34]{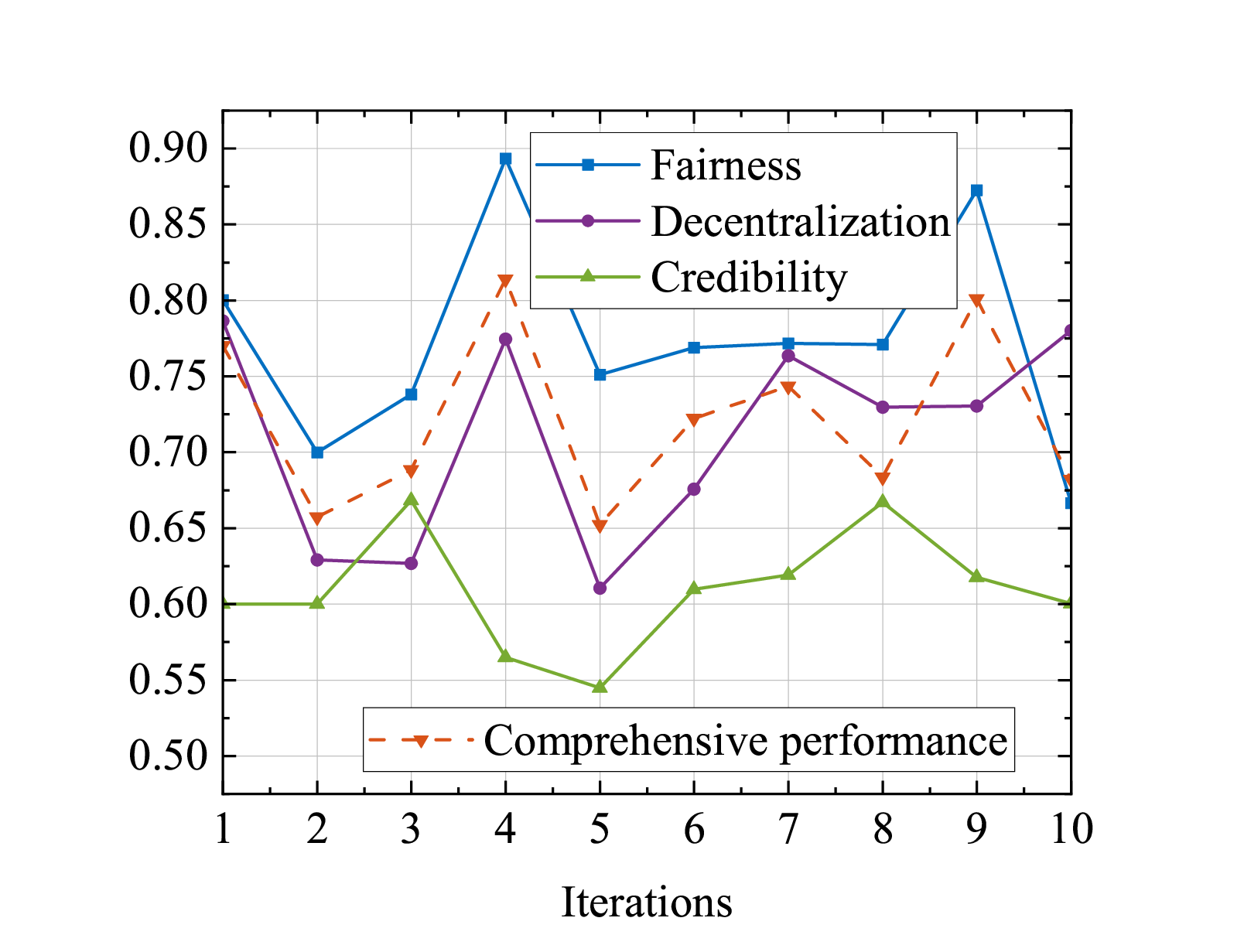}%
\label{fig_5_case2}}
\caption{\small{The performance of the PoPT consensus mechanism. (a)The comprehensive performance  according to definition in Eq. \eqref{O}. (b)The fairness, decentralization, credibility, and the comprehensive performance in the smart grid trading system.}}
\end{figure*}
\begin{figure*}[htbp]
\centering
\subfloat[]{\includegraphics[scale=0.34]{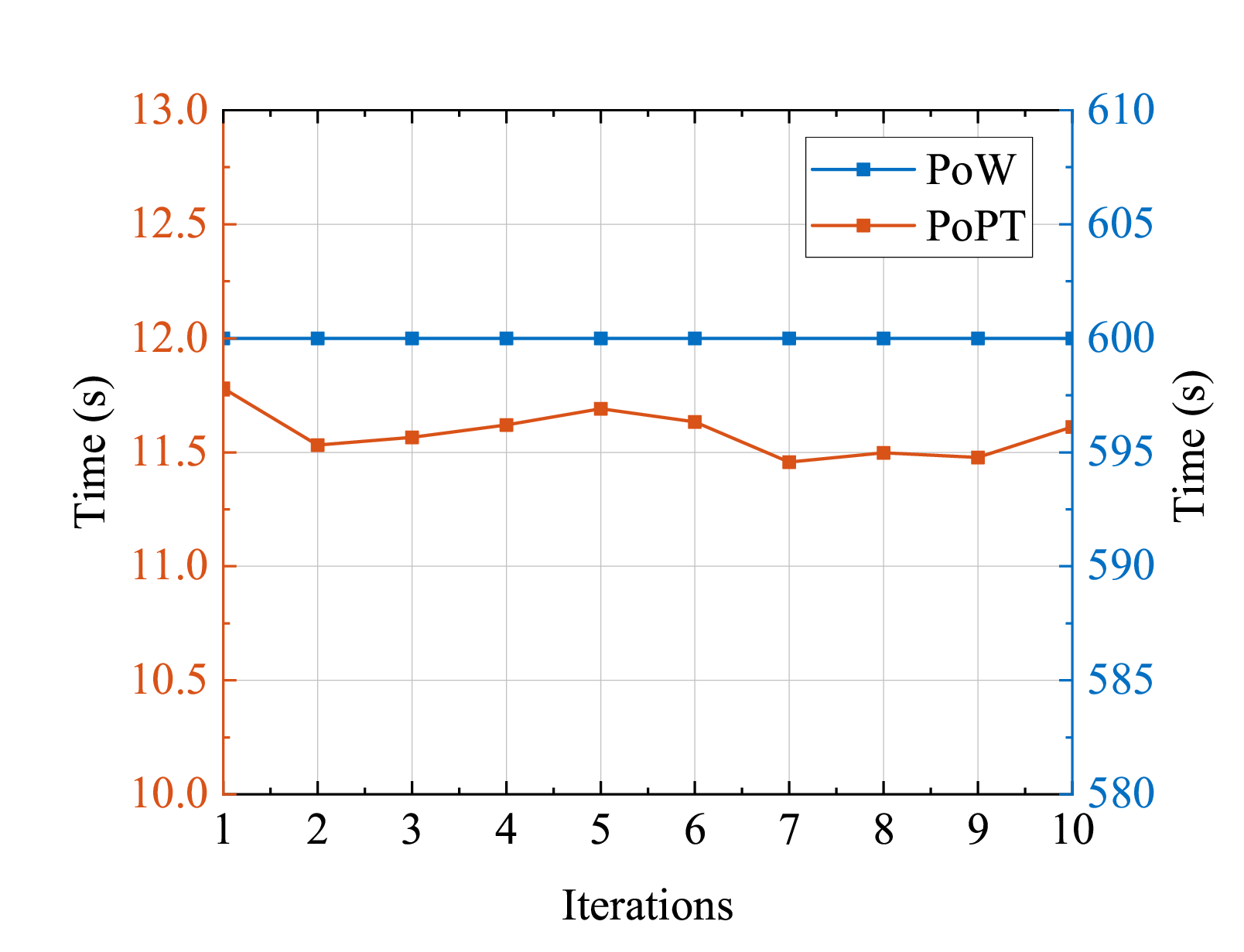}%
\label{fig_6_case1}}
\subfloat[]{\includegraphics[scale=0.34]{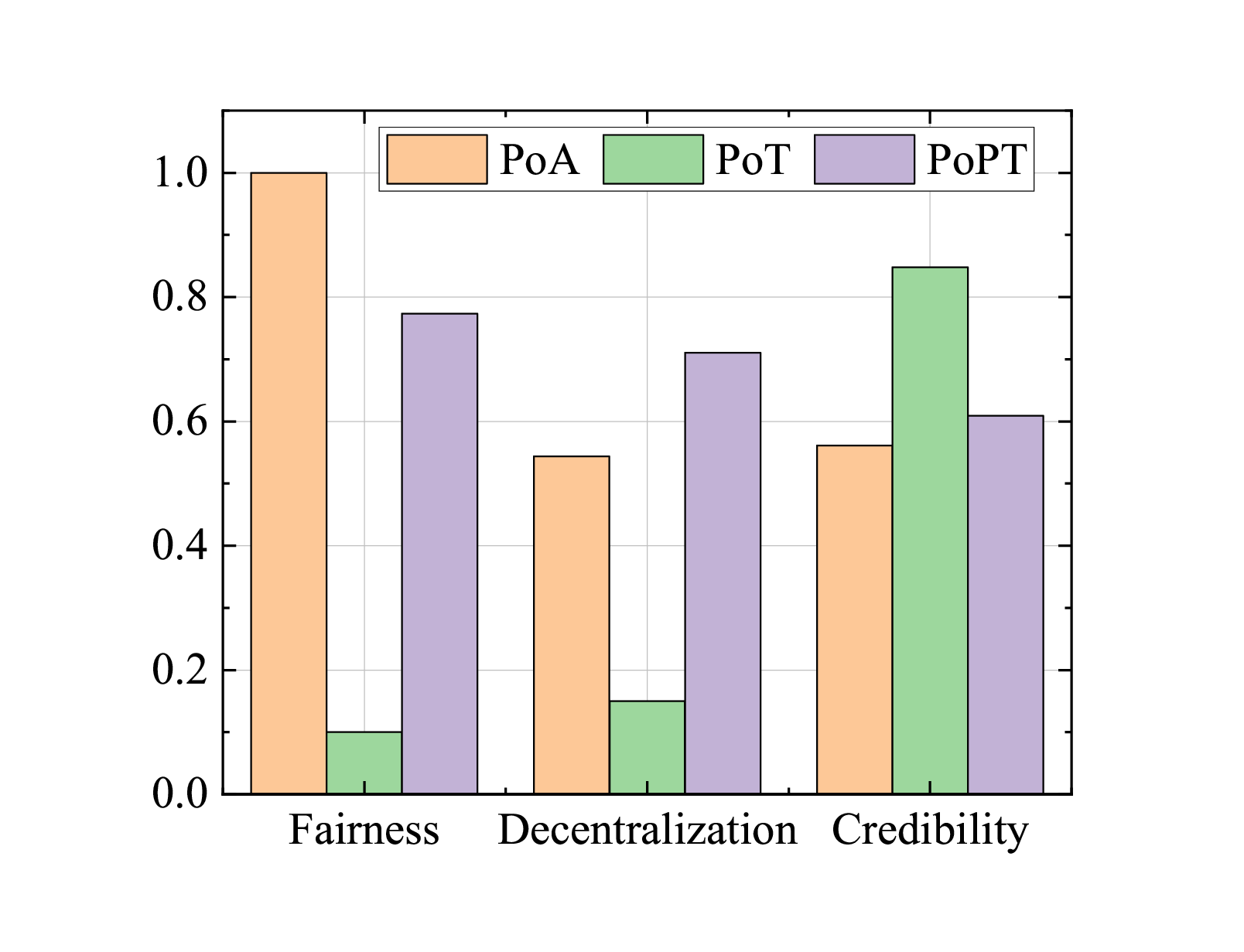}%
\label{fig_6_case2}}
\caption{\small{The comparison of PoPT with other consensus mechanisms. (a)The comparison of time consumption between PoW and PoPT. (b)The comparison of fairness, decentralization, and credibility among PoA, PoT, and PoPT.}}
\end{figure*}
We take the energy trade in a smart grid system as an example.
We specify the options mentioned in subsection \ref{Preliminaries} as choosing electricity trading users and define users as prosumers of electricity who can dispatch energy in a smart grid.
Assume that the price of a transaction is related to the bids of both sellers and buyers.
To ensure the utility of the sellers, the buyers' bids cannot be lower than the national standard electricity price which is set as $0.5$ yuan/kW·h.
$100$ ordinary nodes apply to be block-recorder, i.e., $N_{a}=100$.
The parameters of prospect theory are $\alpha=\beta=0.88$, $\lambda=2.25$, and $\phi=0.74$.
\subsection{Simulation results}
Firstly, we illustrate the effect of seller's personality (referring to seller's Reference Bid (RB) and $\lambda$) on buyers' PV at different bids in Fig. \ref{fig_2_case1}, where RB means the seller's expected price and $\lambda$ indicates sellers sensitivity to loss.
It is shown that the higher the seller's RB or the $\lambda$ is, the lower the buyer's PV is.
While the RB and $\lambda$ are fixed, the higher the buyer's bid is, the higher the PV is.
To further study the influence of personality on PV, we then illustrate the comparison of PV with three types of buyers in Fig. \ref{fig_2_case2}.
We fix one seller, i.e., RB$=0.8$ yuan/kW·h, $\lambda=2.25$.
Let the bids of buyers as $[0.5, 0.7]$, $[0.7, 1.0]$, and $[1.0, 1.2]$, which are denoted as type A, B, and C, respectively.
The results show that during $50$ iterations, the PV of C is always the highest among the three types of buyers, followed by the PV of buyers B, and the PV of buyers A is the lowest.
That is, buyers with higher bids get higher PV in the long run.
This shows that the PoPT consensus mechanism can ensure security to a certain extent.
Because malicious buyers profit from malicious behavior, they offer low bid in their transactions, and low bid leads to low PV in the PoPT consensus mechanism.

Then, we describe the variation of the node's willingness and the optimal reward with the expected utility of the ordinary node in the consensus process in Fig. \ref{fig_3}.
The willingness of a node to be an applicant decreases with increasing the expected utility of an ordinary node, while the optimal reward increase with increasing the expected utility of an ordinary node proportionally.
The higher the expected utility of the ordinary node, the ordinary node is more involuntary to be an applicant node.
This phenomenon leads to the need for higher rewards to attract ordinary nodes.

Moreover, we present $10$ applicant nodes' probability of becoming the block-recorder in Fig. \ref{fig_4} and compare the probabilities with their PV shares.
As defined in fairness, the closer the probability and share, the fairer the PoPT consensus mechanism is.
We observe from Fig. \ref{fig_4} that the difference between the share and the probability is no more than $0.1$ except for applicant node $7$.
As mentioned earlier, the minimum number of applicant nodes that make the sum of probabilities greater than $50\%$ measures the decentralization.
To ensure the decentralization of the PoPT consensus mechanism, most applicant nodes' probability of being block-recorder is slightly less than their PV share.
Since the probability that all applicant nodes become the block-recorder sums up to $1$, luckily for applicant node $7$, its probability of becoming the block-recorder is greater than its PV share.

Afterwards, because the weight of fairness, decentralization, and credibility sums up to $1$, we illustrate the weight of fairness and decentralization to the comprehensive performance of the PoPT consensus mechanism in Eq. \eqref{O} when $\mathcal{F}=0.9$, $\mathcal{D}=0.5$, and $\mathcal{C}=0.1$.
It is shown in Fig. \ref{fig_5_case1} that the weight of fairness has greater effects on comprehensive performance than the weight of decentralization does, which results from fairness being far greater than decentralization and credibility.
We also illustrate the comprehensive performance of PoPT under the environment of the smart grid trading system in Fig. \ref{fig_5_case2}, where fairness, decentralization, and credibility have the same weight.
It shows that the fairness, decentralization, and credibility of PoPT are around $0.75$, $0.70$, and $0.6$, respectively, which means the GWO algorithm is suitable for PoPT to find the optimal probability.
In addition, the comprehensive performance is also considerable.

Finally, we compare PoPT with other consensus mechanisms.
On one hand, we compare the time consumption of PoPT with that of PoW in Fig. \ref{fig_6_case1}.
The hash puzzle makes the average block interval time of PoW up to $600$ seconds, while PoPT only spends about $12$ seconds to determine the block-recorder.
On the other hand, the fairness, decentralization, and credibility of PoPT, PoA \cite{poa}, and PoT \cite{pot} are evaluated in Fig. \ref{fig_6_case2}.
In PoA, nodes have the same chance to be block-recorder regardless of attributes, and PoT only allows nodes with high reputations to be block-recorder.
While the PoPT sets a reasonable probability, which qualifies specific nodes to be block-recorder.
That is, PoA achieves high fairness and decentralization but low credibility, while PoT achieves high credibility but low fairness and decentralization.
Conversely, PoPT achieves good decentralization, higher fairness than PoT, and higher reliability than PoA by probability determination.

\section{Conclusions}\label{conclusion}
In this paper, a game-based consensus mechanism named PoPT concerning the prospect theory is proposed.
PoPT analyzes the consensus process from a game-theoretic perspective and takes decision-making in risk into account.
Moreover, PoPT describes players' psychological activities when facing risky decisions and integrates them into the consensus processes of the blockchain.
The principles and workflows of PoPT are described in detail, and the effectiveness of PoPT is illustrated by a smart grid system.

Different from previous studies, the PoPT consensus mechanism implements game theory to build a behavioral framework in a decentralized environment, which has excellent performance on both fundamental and frontier evaluation metrics.
Regarding design philosophy, the PoPT consensus mechanism integrates prospect theory into the user's interaction, reflecting the user's decision-making.
Regarding performance, the PoPT consensus mechanism does not need to solve hashing problems as PoW does, which saves computing power costs.
Moreover, the PoPT consensus mechanism optimizes the probability of the applicant node being block-recorder and the block reward, which guarantees fairness, decentralization, reliability, and attractiveness.
\bibliography{refabrv}
\bibliographystyle{IEEEtran}

\vfill

\end{document}